\documentclass[11pt]{article}

\usepackage[utf8]{inputenc}
\usepackage[english]{babel}

 \usepackage[pagebackref]{hyperref}
\usepackage{hyperref}
\usepackage{amsmath,amssymb}
\usepackage{enumerate}
\usepackage{amsthm}
\usepackage[all,2cell]{xy}

\usepackage{geometry}
\newtheorem{theorem}{Theorem}

\newtheorem{proposition}{Proposition}

\theoremstyle{definition}
\newtheorem{definition}[theorem]{Definition}

\def\beq{\begin{equation}}
\def\eeq{\end{equation}}
\def\bea{\begin{eqnarray}}
\def\eea{\end{eqnarray}}
\def\beann{\begin{eqnarray*}}
\def\eeann{\end{eqnarray*}}
\def\ben{\begin{enumerate}}
\def\een{\end{enumerate}}
\def\bit{\begin{itemize}}
\def\eit{\end{itemize}}

\def\derpar#1#2{\frac{\partial{#1}}{\partial{#2}}}
\newcommand\restr[2]{{
  \left.\kern-\nulldelimiterspace 
  #1 
  \right|_{#2} 
}}

\newcommand{\R}{\mathbb{R}}

\newcommand{\diff}{\mathrm{d}}

\newcommand{\U}{\mathcal{U}}
\newcommand{\X}{\mathfrak{X}}
\newcommand{\Reeb}{\mathcal{R}}
\newcommand{\calX}{\mathcal{X}}
\newcommand{\calZ}{\mathcal{Z}}

\def\d{\mathrm{d}}

\title{\sf
Constraint algorithm for singular field theories\\
in the $k$-cosymplectic framework}

\author{\sffamily 
Xavier Gr\`acia\thanks{xavier.gracia@upc.edu}, 
Xavier Rivas\thanks{xavier.rivas@upc.edu} and 
Narciso Rom\'an-Roy\thanks{narciso.roman@upc.edu}\\[1ex]
\normalsize\itshape\sffamily 
Department of Mathematics\\ 
\normalsize\itshape\sffamily 
Universitat Polit\`ecnica de Catalunya\\
\normalsize\itshape\sffamily 
Barcelona, Catalonia, Spain}

\date{\sffamily 
19 December 2018}

\begin{document}

\maketitle

\subsection*{Abstract}

The aim of this paper is to develop a constraint algorithm for singular classical field theories in the framework of $k$-cosymplectic geometry.
Since these field theories are singular,
we need to introduce the notion of $k$-precosymplectic structure,
which is a generalization of the $k$-cosymplectic structure.
Next $k$-precosymplectic Hamiltonian systems are introduced
in order to describe singular field theories,
both in Lagrangian and Hamiltonian formalisms.
Finally, we develop a constraint algorithm in order to find a submanifold where the existence of solutions of the field equations is ensured.
The case of affine Lagrangians is studied as a relevant example.

\paragraph*{Keywords:}
$k$-cosymplectic manifold, $k$-precosymplectic manifold, constraint algorithm, singular field theories, Hamiltonian formalism, Lagrangian formalism, affine Lagrangian.
\paragraph*{MSC\,2010:}
Primary: 70H45; Secondary: 35R01, 53C15, 53D99, 70G45, 70S05.

\section{Introduction}

Many theories in modern physics can be formulated
using the tools of differential geometry.
The natural framework for autonomous Hamiltonian mechanical systems is symplectic geometry
\cite{abr78},
whereas its nonautonomous counterpart can be nicely described using cosymplectic or contact geometry
\cite{CLM91,abr78}.
These two formulations admit straightforward generalizations to first order classical field theory using $k$-symplectic and $k$-cosymplectic structures,
which are the generalization to field theories of the autonomous and nonautonomous cases in mechanics
\cite{awa92,LeSaVi2016,Gu-87}.
A more general framework for classical field theories can be built up by using multisymplectic geometry
(see
\cite{nrr2009}
and references therein;
see also \cite{RRSV2011} for an analysis of the relationship among these formulations).

Singular systems are important because of their role in modern physics,
in mechanics and especially in field theory.
In fact, some of the most important physical theories are singular;
for instance, Maxwell's theory of electromagnetism, general relativity, string theory and, in general, all gauge theories.
The main problem of singular theories is the failure of usual existence and uniqueness theorems for the solutions of the differential equations which describe them.
This problem is usually solved by finding a submanifold of the phase space manifold of the system where the existence of solutions is ensured.
This can be done by applying the so-called constraint algorithms.

P.G. Bergmann and P.A.M. Dirac were the first to develop a constraint algorithm to solve the problem for the Hamiltonian formalism of singular mechanics
\cite{AB-51,Di-50}.
These works were written using a local coordinate language and they were later generalized to other situations
(see, for instance,
\cite{BGPR-86,SM-74,Su-82}).
Many people worked in the geometric version of this algorithm, both for the Hamiltonian and Lagrangian formalisms.
Some of the most relevant contributions in this way are
\cite{got79,got78,GP-92,MMT-97,MR-92,Vi-2000},
which dealt with several geometric formulations of autonomous mechanics.
This was later generalized to nonautonomous systems \cite{chi94,dLe02,GM2005}.
These constraint algorithms were adapted to singular field theories in the multisymplectic
\cite{dLe96B,dLe05}
and the $k$-symplectic
\cite{gra09}
frameworks.

The $k$-symplectic formulation, in a certain sense, corresponds to autonomous mechanics.
The non-autonomous analogue of it, namely,
field theories where the Lagrangian or the Hamiltonian functions depend on the coordinates of the base manifold,
is provided by the $k$-cosymplectic formulation.
The aim of this article is to complete this program by developing a constraint algorithm for singular field theories in the framework of $k$-cosymplectic geometry.
Since these field theories are singular,
we need to introduce the notion of $k$-precosymplectic structure,
which is a generalization of the $k$-cosymplectic structure,
and also define $k$-precosymplectic Hamiltonian systems.
Then we will develop a constraint algorithm, similar to those mentioned above, in order to find a constraint submanifold where the existence of solutions to the field equations is ensured.

The paper is organized as follows:
Section~2 is devoted to review several preliminary concepts.
In particular, first we introduce $k$-vector fields and their integral sections, and
next we review the main features about $k$-cosymplectic geometry.
In Section~3, $k$-precosymplectic manifolds are introduced; they are the model of the phase spaces for $k$-cosymplectic field theories described by singular Lagrangians.
We define the concept of a $k$-precosymplectic manifold, we introduce Darboux coordinates in these manifolds, and
we prove the existence of Reeb vector fields for them and discuss conditions for their uniqueness.
These structures are used in Section~4
to present the $k$-cosymplectic formulation
of nonautonomous classical field theory,
both in the Lagrangian and the Hamiltonian formalisms.
Section~5
is devoted to present the constraint algorithm
for $k$-precosymplectic field theories,
which is a generalization of
the algorithm for $k$-presymplectic field theories developed in
\cite{gra09}.
Finally, in Section~6 some examples are discussed:
first, the Lagrangian and Hamiltonian formalisms for the case of affine Lagrangians in general (including a particular model) and second,
a simple model derived from the vibrating string,
both in the Lagrangian and Hamiltonian formalisms.

Throughout the paper all the manifolds and mappings are assumed to be smooth.
Sum over crossed repeated indices is understood.

\section{Preliminaries: \texorpdfstring{$k$}--cosymplectic geometry}

In this section we review the notions of
$k$-vector field and its integral sections,
as well as some general concepts on
$k$-cosymplectic geometry.
Some references on these topics are
\cite{LeSaVi2016,dLe98,dLe01}.			

\subsection{\texorpdfstring{$k$}--vector fields and integral sections}

Let $M$ be an $m$-dimensional smooth manifold and its tangent bundle $\tau\colon TM\to M$.
The \textbf{tangent bundle of $k^1$-velocities} is defined as the Whitney sum
$T^1_kM = TM\oplus_M\overset{k}{\dotsb}\oplus_M TM$,
with the canonical projection $\tau^k\colon T^1_kM\to M$.
			
\begin{definition}
A \textbf{$k$-vector field} $\calX$ on $M$ is a section of the projection $\tau^k$. We denote by $\X^k(M)$ the set of all $k$-vector fields on $M$.
\end{definition}
			
Notice that using the diagram
\begin{equation*}
	\xymatrix{ && T^1_kM \ar[dd]^{\tau^{k,\alpha}} \\\\
	M \ar[uurr]^{\calX} \ar[rr]^{X_\alpha} && TM }
\end{equation*}
every $k$-vector field $\calX$  can be decomposed  as $\calX = (X_1,\dotsc,X_k)$, where $X_\alpha \in\X(M)$.
			
\begin{definition}
Let $\calX=(X_1,\dotsc,X_k)$ be a $k$-vector field on $M$.
An \textbf{integral section} of $\calX$ passing through $p\in M$
is a map $\varphi\colon \U\subset\R^k\to M$,
with $0\in\U$, and such that
\begin{enumerate}[(1)]
\item
$\varphi(0) = p$,
\item
$\displaystyle
T_x\varphi \Big(\frac{\partial}{\partial x^\alpha}\Big|_x\Big) =
X_\alpha(\varphi(x))$,
for every $x\in\U$, for all
$1\leq\alpha\leq k$, and where $\{ x^\alpha\}$
are the canonical coordinates in $\R^k$.
\end{enumerate}
\end{definition}

\subsection{\texorpdfstring{$k$}--cosymplectic geometry}
			
\begin{definition}\label{dfn-k-cosymplectic-manifold}
\cite{LeSaVi2016}
Let $M$ be a manifold of dimension $m = k(n+1)+n$.
A \textbf{$k$-cosymplectic structure} on $M$ is a family $(\eta^\alpha,\omega^\alpha,V;1\leq\alpha\leq k)$,
where each $\eta^\alpha$ is a closed 1-form,
each $\omega^\alpha$ is a closed 2-form and
$V$ is an integrable $nk$-dimensional distribution on $M$ satisfying
\begin{enumerate}[(1)]
\item $\eta^1\wedge\dotsb\wedge\eta^k\neq 0$, $\restr{\eta^\alpha}{V} = 0$, $\restr{\omega^\alpha}{V\times V} = 0$,
\item $\left( \bigcap_{\alpha = 1}^k \ker\eta^\alpha \right)\cap\left( \bigcap_{\alpha = 1}^k \ker\omega^\alpha \right) = \{0\}$, $\dim\left( \bigcap_{\alpha = 1}^k \ker\omega^\alpha \right) = k$.
\end{enumerate}
Then, $(M,\eta^\alpha,\omega^\alpha,V)$ is said to be a \textbf{$k$-cosymplectic manifold}.
\end{definition}

In particular, if $k=1$, then $\dim M = 2n+1$ and $(\eta^1,\omega^1)$ is a cosymplectic structure on $M$.

\begin{definition}
Let $(M,\eta^\alpha,\omega^\alpha,V)$ be a $k$-cosymplectic manifold.
Then there exists a family of $k$ vector fields $\{\Reeb_\alpha\}$,
which are called \textbf{Reeb vector fields},
characterized by the following conditions
\begin{equation}
\label{eq-Reeb}
	i_{\Reeb_\alpha}\eta^\beta = \delta^\beta_\alpha \,,\quad i_{\Reeb_\alpha}\omega^\beta = 0 \,.
\end{equation}
\end{definition}

\begin{theorem}   [Darboux theorem for $k$-cosymplectic manifolds]
Let $(M,\eta^\alpha,\omega^\alpha,V)$ be a $k$-cosymplectic manifold.
Then around each point of $M$ there exist local coordinates
$(x^\alpha, y^i,y_i^\alpha)$ with $1\leq\alpha\leq k,1\leq i\leq n$
such that
\begin{equation*}
	\eta^\alpha = \d x^\alpha \,,\quad
	\omega^\alpha = \d y^i\wedge\d y_i^\alpha,\quad
	V= \left\langle
	  \frac{\partial}{\partial y^1_i} \,, \dotsc,    \frac{\partial}{\partial y^k_i}
	\right\rangle_{i = 1,\dotsc,n} .
\end{equation*}
These are the so-called \textbf{Darboux} or \textbf{canonical coordinates} of the $k$-cosymplectic manifold $M$.
\end{theorem}   
\begin{proof}
The proof of this theorem can be found in \cite{dLe98}.
\end{proof}

Given a $k$-cosymplectic manifold $(M,\eta^\alpha,\omega^\alpha,V)$,
we can define two vector bundle morphisms
\begin{equation*}
	\begin{array}{rccl}
	\widetilde\flat\colon & T^1_kM & \longrightarrow & (T^1_k)^\ast M \\
	& \calX & \longmapsto & (i_{X_1}\omega^1 + (i_{X_1}\eta^1)\eta^1,\dotsc,i_{X_k}\omega^k + (i_{X_k}\eta^k)\eta^k)
	\end{array}
\end{equation*}
and
\begin{equation*}
	\begin{array}{rccl}
	\flat\colon & T^1_kM & \longrightarrow & T^\ast M \\
	& \calX & \longmapsto & i_{X_\alpha}\omega^\alpha + (i_{X_\alpha}\eta^\alpha)\eta^\alpha
	\end{array}
\end{equation*}
			
{\bf Remark:} {\rm
Notice that $\flat = \mathrm{tr}(\widetilde\flat)$, and hence in the case $k=1$ we have that $\flat = \widetilde\flat$ which is the $\flat$ morphism defined for cosymplectic manifolds.}

Taking Darboux coordinates, the Reeb vector fields are
\begin{equation*}
    \mathcal{R}_\alpha = \frac{\partial}{\partial x^\alpha}.
\end{equation*}

\subsection{Trivial \texorpdfstring{$k$}--cosymplectic manifolds}

A trivial example of $k$-cosymplectic manifold
is provided by
the cartesian product of the euclidean space $\R^k$ with a $k$-symplectic manifold.
Remember that a \textbf{$k$-symplectic manifold} is an $n(k+1)$-dimensional differentiable manifold $N$ endowed with a $k$-symplectic structure,
that is, a family
$(\varpi^1,\ldots, \varpi^k,\mathcal{V})$, where $\mathcal{V}$ is
an $nk$-dimensional integrable distribution in $N$ and
$\varpi^1,\ldots,\varpi^k$ are closed differentiable $2$-forms in $N$
satisfying that:
$\varpi^\alpha\Big\vert_{\mathcal{V}\times\mathcal{V}}=0, (1\leq \alpha\leq k)\,$,
and
$\displaystyle\bigcap_{\alpha=1}^k\ker \varpi^\alpha=\{0\}$.
Then, using the canonical projections
\begin{equation*}
	\pi_{\R^k}\colon \R^k\times N\longrightarrow\R^k,\qquad \pi_N\colon\R^k\times N\longrightarrow N
\end{equation*}
we can define differential forms
\begin{equation*}
	\eta^\alpha = \pi_{\R^k}^\ast(\d x^\alpha),\qquad
	\omega^\alpha = \pi_N^\ast\varpi^\alpha,
\end{equation*}
and the distribution $\mathcal{V}$ in $N$ defines a distribution $V$ in $M=\R^k\times N$ in a natural way.
All conditions given in Definition \ref{dfn-k-cosymplectic-manifold} are satisfied, and hence $M=\R^k\times N$ endowed with the $k$-cosymplectic structure $(\eta^\alpha,\omega^\alpha,V)$ is a $k$-cosymplectic manifold.

\medskip

Then,
the simplest model of $k$-cosymplectic manifold is the so called
\textbf{stable cotangent bundle of $k^1$-covelocities}
of an $n$-dimensional manifold $Q$,
denoted as
$\R^k \times (T^1_k)^\ast Q$,
where $(T^1_k)^\ast Q$ is the Whitney sum of $k$ copies of the cotangent bundle of $Q$,
i.e.
$(T^1_k)^\ast Q=
T^\ast Q\oplus_Q\overset{(k)}{\dotsb}\oplus_Q T^\ast Q$
(compare with the definition of the tangent bundle of $k^1$-velocities).
Thus, the elements of $\R^k\times (T^1_k)^\ast Q$ are of the form
$(x,\nu_{1_q},\dotsc,\nu_{k_q})$
where $x\in\R^k$, $q\in Q$ and $\nu_{\alpha_q}\in T^\ast_qQ$
where $1\leq\alpha \leq k$.

In the following diagram we collect the projections we use from now on:
{\small 
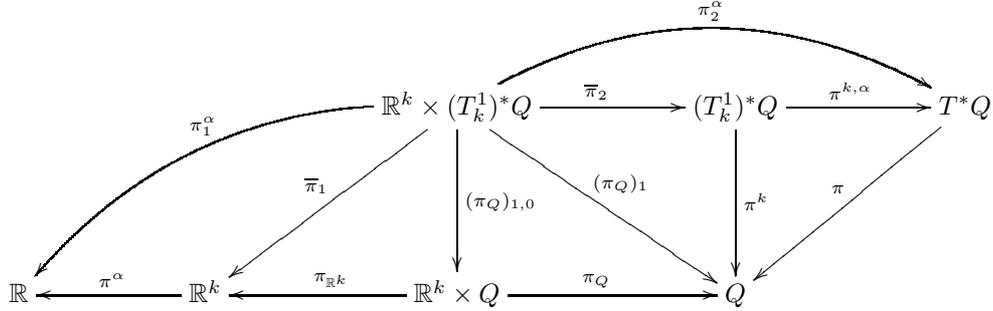
\begin{equation*}
	\xymatrix{ &&&& \R^k\times (T^1_k)^\ast Q \ar[rr]^{\overline{\pi}_2} \ar[ddrr]^{(\pi_Q)_1} \ar[dd]^{(\pi_Q)_{1,0}} \ar[ddll]_{\overline{\pi}_1} \ar@/_2pc/[ddllll]_{\pi_1^\alpha} \ar@/^2.5pc/[rrrr]^{\pi_2^\alpha} && (T^1_k)^\ast Q \ar[rr]^{\pi^{k,\alpha}} \ar[dd]^{\pi^k} && T^\ast Q \ar[ddll]_{\pi} \\\\
	\R && \R^k \ar[ll]_{\pi^\alpha} && \R^k\times Q \ar[ll]_{\pi_{\R^k}} \ar[rr]^{\pi_Q} && Q}
\end{equation*}}%

If $(q^i)$, with $1\leq i\leq n$,
is a local coordinate system defined on an open set $U\subset Q$,
the induced \textbf{local coordinates}
$(x^\alpha,q^i,p_i^\alpha)$,
$1\leq i \leq n$, $1\leq \alpha\leq k$
on $\R^k\times (T^1_k)^\ast U = \left((\pi_Q)_1\right)^{-1}(U)$
are given by
\begin{align*}
	x^\alpha(x,\nu_{1_q},\dotsc,\nu_{k_q}) &= x^\alpha(x) = x^\alpha,
	\\
	q^i(x,\nu_{1_q},\dotsc,\nu_{k_q}) &= q^i(q),
	\\
	p_i^\alpha(x,\nu_{1_q},\dotsc,\nu_{k_q}) &=
	\nu_{\alpha_k} \bigg(\left.\frac{\partial}{\partial q^i}\right|_q\bigg).
\end{align*}

Thus,
$\R^k\times (T^1_k)^\ast Q$
is endowed with a $k$-cosymplectic structure and thus
it is a $k$-cosymplectic manifold of dimension $k+n(k+1)$,
which has the structure of a vector bundle over~$Q$
with the projection $(\pi_Q)_1$.

On $\R^k\times (T^1_k)^\ast Q$ we can define a family of canonical forms:
\begin{equation*}
	\eta^\alpha = (\pi_1^\alpha)^\ast\d x,\qquad
	\theta^\alpha = (\pi_2^\alpha)^\ast\theta,\qquad
	\omega^\alpha = (\pi_2^\alpha)^\ast\omega,
\end{equation*}
with $1\leq \alpha \leq k$, being
$\pi_1^\alpha\colon \R^k\times (T^1_k)^\ast Q\to \R$ and
$\pi^\alpha_2\colon \R^k\times (T^1_k)^\ast Q\to T^\ast Q$
the projections defined by
\begin{equation*}
	\pi_1^\alpha(x,\nu_{1_q},\dotsc,\nu_{k_q}) = x^\alpha,\qquad
	\pi_2^\alpha(x,\nu_{1_q},\dotsc,\nu_{k_q}) = \nu_{\alpha_q}
\end{equation*}
and $\theta$ and $\omega$ are the canonical Liouville and symplectic forms on $T^\ast Q$, respectively.
Observe that, since $\omega = -\d\theta$,
then $\omega^\alpha = -\d\theta^\alpha$.

If we consider a local coordinate system $(x^\alpha,q^i,p^\alpha_i)$ on $\R^k\times (T^1_k)^\ast Q$, the \textbf{canonical forms} $\eta^\alpha$, $\theta^\alpha$ and $\omega^\alpha$ have the following local expressions:
\begin{equation*}
	\eta^\alpha = \d x^\alpha,\qquad \theta^\alpha = p_i^\alpha\d q^i,\qquad \omega^\alpha = \d q^i\wedge\d p^\alpha_i\ .
\end{equation*}
Moreover, let $V= \ker T (\pi_Q)_{1,0}$.
In local coordinates,
the forms $\eta^\alpha$ and $\omega^\alpha$ are closed,
and the following relations hold:
\begin{enumerate}[(1)]
\item
$\d x^1\wedge\dotsb\wedge\d x^k \neq 0$,
$\restr{\d x^\alpha}{V} = 0$,
$\restr{\omega^\alpha}{V\times V} = 0$,
\item
$\left( \bigcap_{\alpha=1}^k\ker\d x^\alpha \right)
\cap
\left( \bigcap_{\alpha = 1}^k\ker \omega^\alpha \right)
=
\{0\}$,
$\dim\left( \bigcap_{\alpha = 1}^k\ker \omega^\alpha \right) = k$.
\end{enumerate}

{\bf Remark:} {\rm
Notice that the canonical forms
on $(T^1_k)^\ast Q$ and $\R^k\times (T^1_k)^\ast Q$
are $(\overline{\pi}_2)^\ast$-related.}

\section{\texorpdfstring{$k$}--precosymplectic manifolds}
\label{kpresec}

In the same way as
$k$-presymplectic manifolds generalize $k$-symplectic manifolds,
$k$-precosymplectic manifolds are a generalization of $k$-cosymplectic manifolds
when some degeneracy is accepted in the 2-forms of the structure.
		
\begin{definition}
\label{dfn-singular-k-co}
Let $M$ be a differentiable manifold  of dimension $k(n+1)+n-\ell$
(with  $1\leq\ell\leq nk$). A
\textbf{$k$-precosymplectic structure} in $M$ is a family
$(\eta^\alpha,\Omega^\alpha,V)$, $1\leq \alpha\leq k$, where $\eta^\alpha$ are closed
$1$-forms in $M$, $\omega^\alpha$ are closed $2$-forms in $M$ such that
${\rm rank}\,\omega^\alpha=2r_\alpha$, with $1\leq r_\alpha\leq n$,
and $V$ is an integrable $nk$-dimensional distribution in $M$ satisfying that:
\begin{enumerate}[(1)]
\item
$\eta^1\wedge\dots\wedge \eta^k\neq 0$,\quad $\eta^\alpha\vert_V=0,\quad \omega^\alpha\vert_{V\times V}=0,$
\item
$\displaystyle \dim\,\Big({\bigcap_{\alpha=1}^{k}}\ker \omega^\alpha_p\Big)\geq k \ \ \mbox{\rm (for every $p\in M$)}  \ .$
\end{enumerate}
A manifold $M$ endowed with a $k$-precosymplectic structure is said to be a
\textbf{$k$-precosymplectic manifold}.
\end{definition}

In particular, if $k=1$, then $\dim M = 2n+1-\ell$ and $(M,\eta^1,\omega^1)$ is a precosymplectic manifold as is defined in \cite{IM-92}, and the so-called gauge distribution is given by $\ker\omega^1\cap\ker\eta^1$.

\noindent\textbf{Example} \
As in the regular case,
a simple example of $k$-precosymplectic manifold
can be constructed
from a $k$-presymplectic manifold.
Recall that a {\bf $k$-presymplectic manifold} is a family $(P, \varpi^\alpha, \mathcal{V})$ where $\varpi^\alpha$ are closed 2-forms in $P$ and $\mathcal{V}$ is a $nk$-dimensional integrable distribution satisfying $\restr{\varpi^\alpha}{\mathcal{V}\times\mathcal{V}} = 0$ for every $1\leq\alpha\leq k$.

Under these hypothesis, the product manifold
$\R^k\times P$
is a $k$-precosymplectic manifold taking
$\eta^\alpha = \tau^\ast\d t^\alpha$
where $t^\alpha$ are the canonical coordinates in $\R^k$ and $\tau$ is the canonical projection
$\R^k\times P\overset{\tau}{\longrightarrow}\R^k$
and $\omega^\alpha = \pi^\ast\varpi^\alpha$
where $\pi$ is the canonical projection $\R^k\times P\overset{\pi}{\longrightarrow}P$.
In the description of the algorithm,
we will ask our manifolds to be of this type
in order to have the problem well defined.

\medskip
			
In Definition \ref{dfn-singular-k-co} we have imposed the condition of the existence of a distribution $V$ because it is precisely the existence of this distribution what ensures the existence of Darboux coordinates in the regular case.
It is still an open problem to characterize the conditions for their existence in the singular case.
However, from now on we will assume the existence of Darboux coordinates around every point (as it happens, for instance, in the previous example).
In more detail,
let  $M$ be a  $k$-precosymplectic manifold such that
${\rm rank}\,\omega^\alpha=2r_\alpha$, with $1\leq r_\alpha\leq n$ and
$\displaystyle d=kn-\sum_{\alpha=1}^kr_\alpha-\ell$;
around every point $p\in M$, we assume the existence of a local chart of coordinates
$$
(\mathcal{U}_p;x^\alpha,y^i,y^\alpha_{i_\alpha},z^j)\quad ;\quad
1\leq\alpha\leq k  \ ,\  1\leq i\leq n\ ,\
 i_\alpha\in I_\alpha\subseteq\{1,\dots,n\}\ ,\ 1\leq j\leq d\ ,
$$
 such that
\begin{eqnarray*}
\eta^\alpha\vert_{\mathcal{U}_p}=
\d x^\alpha
,&&\quad
\omega^\alpha\vert_{\mathcal{U}_p}=
\d y^{i_\alpha}\wedge\,\d y^\alpha_{i_\alpha}
\\
\displaystyle
V\vert_{\mathcal{U}_p}=
\left\langle\frac{\partial}{\partial y^\alpha_{i_\alpha}},
\frac{\partial}{\partial z^j}\right\rangle
,&&
\quad\displaystyle
\left[ ({\bigcap_{\alpha=1}^{k}}\ker \eta^\alpha) \cap
({\bigcap_{\alpha=1}^{k}}\ker\omega^\alpha) \right]
\big\vert_{\mathcal{U}_p}=
\left\langle\frac{\partial}{\partial z^j}\right\rangle \ .
\end{eqnarray*}

To discuss Hamilton's equations we will need the Reeb vector fields
$\mathcal{R}_\alpha$, defined by
Eq.~\eqref{eq-Reeb}.
We already mentioned that they are unique in the
$k$-cosymplectic case.
Now we are going to prove their existence in the singular case,
although they will not be unique.
			
\begin{proposition}
    Given a $k$-precosymplectic manifold $(M,\omega^\alpha,\eta^\alpha,V)$ with Darboux charts,
    there exists a family $Y_1,\dotsc,Y_k\in\X(M)$ of vector fields satisfying
    \begin{equation*}
    	\begin{cases}
    	i_{Y_\alpha}\omega^\beta = 0,\\
    	i_{Y_\alpha}\eta^\beta = \delta_\alpha^\beta.
    	\end{cases}
    \end{equation*}
\end{proposition}
\begin{proof}
    Consider a partition of unity
    $\{(\U_\lambda,\psi_\lambda)\}_{\lambda\in\Lambda}$
    on $M$ such that on every $\U_b$ we have Darboux coordinates
    $\{x^\alpha_\lambda,y^i_\lambda,y^\alpha_{i_\alpha,\lambda};z^j_\lambda\}$.
    Consider now the local vector fields
    $\displaystyle Y^\lambda_\alpha = \frac{\partial}{\,\partial x^\alpha_\lambda}$.
    These vector fields satisfy
    \begin{equation*}
    	\begin{cases}
    	i_{Y^\lambda_\alpha}\omega^\beta = 0,\\
    	i_{Y^\lambda_\alpha}\eta^\beta = \delta^\beta_\alpha
    	\end{cases}
    \end{equation*}
    on $\U_\lambda$.
    Using these vector fields, we can define global vector fields
    \begin{equation*}
    	\widetilde{Y}^\lambda_\alpha(p) =
    	\begin{cases}
    	\psi_\lambda(p)Y^\lambda_\alpha(p) & \mbox{if }p\in \U_\lambda,\\
    	0 & \mbox{if }p\notin\U_\lambda.
    	\end{cases}
    \end{equation*}
    With these global vector fields we can construct global vector fields
    $Y_\alpha = \sum_{\lambda\in\Lambda} \widetilde{Y}^\lambda_\alpha$
    which satisfy
    \begin{equation*}
    	\begin{cases}
    	i_{Y_\alpha}\omega^\beta = 0,\\
    	i_{Y_\alpha}\eta^\beta = \delta^\beta_\alpha,
    	\end{cases}
    \end{equation*}
    for every $\alpha,\beta = 1,\dotsc,k$.
\end{proof}
			
The vector fields provided by this proposition are not necessarily unique. In fact, the Reeb vector fields can be written in Darboux coordinates as
\begin{equation*}
    \mathcal{R}_\alpha = \frac{\partial}{\partial x^\alpha} + D^j_\alpha\frac{\partial}{\partial z^j}
\end{equation*}
for arbitrary coefficients $D^j_\alpha$.

{\bf Remark:} 
\label{choice-reeb}
\rm
Nevertheless,
sometimes one can impose some extra conditions
that determine the Reeb vector fields uniquely.
Consider for instance the situation where
the $k$-precosymplectic manifold $M$ is of the type
$\R^k\times P$, where $P$ is a $k$-presymplectic manifold.

The canonical vector fields
$\frac{\partial}{\partial x^\alpha}$ of $\R^k$
can be lifted canonically to vector fields on the product $\R^k \times P$.
These vector fields are denoted also by $\frac{\partial}{\partial x^\alpha}$
and are a family of Reeb vector fields
of the $k$-precosymplectic manifold $\R^k \times M$.

\section{\texorpdfstring{$k$}--cosymplectic formulation of nonautonomous field theories}
\label{section-k-cosymplectic}
	
The $k$-cosymplectic formulation allows to describe field theories
where the Lagrangian or the Hamiltonian functions
depend explicitly on the coordinates of the basis
(space-time coordinates or similar).
Therefore it is a generalization of the $k$-symplectic formulation,
where these coordinates do not appear explicitly
\cite{LeSaVi2016}.
It is also the generalization of the standard cosymplectic formalism for non-autonomous mechanics
\cite{CLM91}.

Next we review the main features of the Lagrangian and Hamiltonian formalisms in this formulation
(see \cite{LeSaVi2016,dLe98,dLe01,RRSV-2011} for details).

\subsection{\texorpdfstring{$k$}--cosymplectic Hamiltonian formalism}
\label{kHamForm}
			
\begin{definition}
Let $(M,\omega^\alpha,\eta^\alpha,V)$ be a $k$-cosymplectic manifold and let ${\gamma\in\Omega^1(M)}$ be a closed $1$-form on $M$,
which will be called the \textbf{Hamiltonian $1$-form}.
The family $(M,\omega^\alpha,\eta^\alpha,V, \gamma)$ is a \textbf{$k$-cosymplectic Hamiltonian system}.

A $k$-vector field
$\calX = (X_1,\dotsc,X_k)\in\X^k(M)$
is said to be a
\textbf{$k$-cosymplectic Hamiltonian $k$-vector field}
if it is solution to the system of equations
\begin{equation}
\label{eq-Hamilton-k-cosymplectic}
	\begin{cases}
		i_{X_\alpha}\omega^\alpha = \gamma - \gamma(\Reeb_\alpha)\eta^\alpha \\
		i_{X_\beta}\eta^\alpha = \delta^\alpha_\beta.
	\end{cases}
\end{equation}
We use the notation $\calX\in\X^k_h(M)$.
\end{definition}

Notice that when $k=1$ we recover the equation of motion for a cosymplectic Hamiltonian system
\cite{chi94,dLe02}.

As $\gamma$ is a closed $1$-form, by Poincar\'e's Lemma there exists a local function $h$
such that $\gamma=\d h$.

Using the $\flat$ morphism defined in the previous sections, we can write equations \eqref{eq-Hamilton-k-cosymplectic} as
\begin{equation}\label{eq-Hamilton-k-cosymplectic-2}
	\begin{cases}
	\flat(\calX) = \gamma + (1-\gamma(\Reeb_\alpha))\eta^\alpha \\
	i_{X_\beta}\eta^\alpha = \delta^\alpha_\beta.
	\end{cases}
\end{equation}
Consider an arbitrary $k$-vector field $\calX = (X_\alpha)\in\X^k(M)$, which in a canonical chart is expressed as
\begin{equation*}
	X_\alpha =
	(A_\alpha)^\beta \frac{\partial}{\partial x^\beta} +
	(B_\alpha)^i \frac{\partial}{\partial q^i} +
	(C_\alpha)^\beta_i \frac{\partial}{\partial p_i^\beta}, \qquad
	1\leq\alpha\leq k.
\end{equation*}
Imposing equations \eqref{eq-Hamilton-k-cosymplectic}, we get the conditions
\begin{equation}
\label{eq-k-co-coordinates}
	\begin{cases}
	(A_\alpha)^\beta = \delta_\alpha^\beta,\\
	\dfrac{\partial h}{\partial p_i^\alpha} = (B_\alpha)^i,\\
	\dfrac{\partial h}{\partial q^i} =
	  -\sum_{\beta = 1}^k(C_\beta)^\beta_i,
	\end{cases}
\end{equation}
where $1\leq i \leq m$ and $1\leq \alpha\leq k$.
Notice that these conditions do not depend on the choice of the
Reeb vector fields.
However, we need the Reeb vector fields to write the system of equations~\eqref{eq-Hamilton-k-cosymplectic}.
In remark \ref{choice-reeb} we discussed how to choose a family of Reeb vector fields.

If
$\psi:\R^k\to \R^k\times (T^1_k)^{\;*}Q$,
locally given by
$\psi(x)=(\psi^\alpha(x),\psi^i(x),\psi^\alpha_i(x))$,
is an integral section of $\mathcal{X}$,
then, from (\ref{eq-k-co-coordinates}),
we obtain that $\psi$ is a solution to the Hamiltonian field equations
$$
\frac{\partial h}{\partial q^i} =
\sum_{\alpha=1}^k \frac{\partial\psi^\alpha_i}{\partial x^\alpha}
\:, \quad
\frac{\partial h}{\partial p^\alpha_i} =
\frac{\partial\psi^i}{\partial x^\alpha} \:.
$$

\subsection{\texorpdfstring{$k$}--cosymplectic Lagrangian formalism}
\label{kLagForm}

Consider the tangent bundle of $k$-velocities $T^1_kQ$,
with coordinates $(q^i , v^i_\alpha)$.
For a vector $X_q$ at $Q$,
its \textbf{vertical $\alpha$-lift}
$(X_q)^\alpha$ is defined as the vector on $T_k^1Q$ given by
$$
(X_q)^\alpha({v_1}_q,\ldots , {v_k}_q) =
\displaystyle\frac{d}{ds} (
{v_1}_q,\ldots,{v_{\alpha-1}}_q,{v_\alpha}_q+sX_q,{v_{\alpha+1}}_q,
\ldots,{v_k}_q)_{\vert_{s=0}}
\,,
$$
where
$({v_1}_q,\ldots,{v_k}_q) \in T^1_kQ$.
In local coordinates, if $\displaystyle X_q=a^i\derpar{}{q^i}\Big\vert_q$, then
$\displaystyle
(X_q)^\alpha =  a^i \displaystyle\frac{\partial}{\partial
v^i_\alpha}\Big\vert_v$.

The \textbf{canonical $k$-tangent structure} on $T^1_kQ$ is the set
$(J^1,\ldots,J^k)$ of tensor fields  of type $(1,1)$ defined by
$$
J^\alpha(v)(Z_{v}) =
(T_v\tau^Q (Z_{v}))^\alpha
\quad
\mbox{for every}\
Z_{v}\in T_{v}(T^1_kQ)\,,\
v=({v_1}_q,\ldots , {v_k}_q)\in T^1_kQ
\:.
$$
In local  coordinates we have that
$\displaystyle J^\alpha=\displaystyle\frac
{\displaystyle\partial}{\displaystyle\partial v^i_\alpha} \otimes dq^i$.

Furthermore,
we have the \textbf{Liouville vector fields}
$\Delta_\alpha\in{\mathfrak X}(T^1_kQ)$,
which are the infinitesimal generators of the flows
$$
\begin{array}{ccl}
\R \times T^1_kQ & \longrightarrow &
 T^1_kQ  \\
\noalign{\medskip} (s,({v_1}_q,\ldots , {v_k}_q)) &
\longrightarrow & ({v_1}_q,\ldots,{v_{\alpha-1}}_q, e^s \,
{v_\alpha}_q,{v_{\alpha+1}}_q, \ldots,{v_k}_q)\, .
\end{array}
$$
We denote $\Delta=\sum_{\alpha}\Delta_\alpha$.
In local coordinates
$\displaystyle
\Delta_\alpha =
\sum_{i} v^i_\alpha  \frac{\partial}{\,\partial v_\alpha^i}
$.

Consider now the phase space $\R^k\times T^1_kQ$.
The canonical structures $J^\alpha$ and the Liouville vector fields $\Delta_\alpha$
can be trivially extended from $T^1_kQ$ to $\R^k\times T^1_kQ$,
and are denoted also by $J^\alpha$ and $\Delta_\alpha$.
If $(x^\alpha,q^i,v^i_\alpha)$ are the natural coordinates in $\R^k\times T^1_kQ$,
their local expressions are the same as above.
Using them, we can define:
			
\begin{definition}
A $k$-vector field $\calX\in\X^k(\R^k\times T^1_kQ)$ is a
\textbf{second order partial differential equation} (\textsc{sopde})
if it satisfies the following conditions:
\begin{enumerate}[(1)]
\item
$J^\alpha(X_\alpha) = \Delta_\alpha$
for $\alpha$ fixed, with $1\leq\alpha\leq k$;
\item
$i_{X_\beta}\eta^\alpha = \delta^\alpha_\beta$
for every $1\leq\alpha,\beta\leq k$.
\end{enumerate}
\end{definition}

The local expression
of a {\sc sopde} $\mathcal{X}=(X_1 ,\ldots,X_k) $ is
\begin{equation}
\label{localsode2}
X_\alpha=
\frac{\partial}{\partial x^\alpha}+v^i_\alpha\frac{\displaystyle\partial} {\displaystyle
\partial q^i}+(X_\alpha)^i_\beta \frac{\displaystyle\partial} {\displaystyle \partial v^i_\beta}\ .
\end{equation}

\begin{definition}
\label{de652}
Let $\phi\colon\R^k \rightarrow Q$, the
\textbf{first prolongation}
$\phi^{[1]}$ of $\phi$ is the map
$$
\begin{array}{rcl}
\phi^{[1]}\colon\R^k &\longrightarrow &\R^k\times T^1_kQ \\ x &
\longmapsto &
\displaystyle
(x,j^1_0\phi_x) \equiv
\left(x,
T_x\phi \left(\frac{\partial}{\partial x^1} \Big\vert_x\right),
\ldots ,
T_x\phi \left(\frac{\partial}{\partial x^k}\Big\vert_x \right)
\right)
\end{array}
$$
The section $\phi^{[1]}$ is said to be a
\textbf{holonomic section}.
In coordinates
$$
\phi^{[1]}(x^1,\dots, x^k) =
\left( x^1,\dots,x^k,\phi^i(x),
\frac{\partial\phi^i}{\partial x^\alpha}(x) \right) \,.
$$
\end{definition}

\begin{proposition}\label{lem0}
If $\mathcal{X}=(X_1 ,\ldots,X_k) $ is an integrable {\sc sopde}, then a map $\psi:\R^k \to \R^k\times T^1_kQ$, given by
$\psi(x)=(\psi^\alpha(x),\psi^i(x),\psi^i_\alpha(x))$, is an integral section
of $(X_1 ,\ldots,X_k) $ if, and only if, 
\begin{equation}\label{nn}
\psi^\alpha(x)=x^\alpha \, , \quad \psi^i_\alpha(x)=\frac{\displaystyle\partial \psi^i}
{\displaystyle\partial x^\alpha}(x) \, , \quad \frac{\displaystyle\partial^2 \psi^i}
{\displaystyle\partial x^\alpha \displaystyle\partial x^\beta}(x)=(X_\alpha)^i_\beta(\psi(x)) \, .
\end{equation}
In this case, $\psi$ is a holonomic section.
\end{proposition}

Observe that if $\mathcal{X}=(X_1 ,\ldots,X_k)$ is integrable, from
(\ref{nn}) we deduce that $(X_\alpha)^i_\beta=(X_\beta)_\alpha^i$.
	
A \textbf{Lagrangian function}
is a function $L\colon\R^k\times T^1_kQ\to\R$.
From it one can define a family of 1-forms
$\theta^1_L,\dots,\theta^k_L \in \Omega^1(\R^k\times T^1_kQ)$
as
\begin{equation*}
\theta^\alpha_L = \d L\circ J^\alpha,
\end{equation*}
and from these 1-forms one can define the so-called \textbf{Poincar\'e--Cartan} 2-forms
\begin{equation*}
\omega^\alpha_L = -\d\theta^\alpha_L.
\end{equation*}
These forms,
together with the vertical tangent distribution
$V = \ker T\,(\pi_{\R^k})_{1,0}$,
define a $k$-precosymplectic structure
$(\d x^\alpha,\omega^\alpha_L,V)$
on $\R^k\times T^1_kQ$.

\begin{definition}
Let $L$ be a Lagrangian function on $\R^k\times T^1_kQ$.
We say that $L$ is a {\bf regular Lagrangian} if, for every $1\leq\alpha, \beta\leq k$ and every $p\in\R^k\times T^1_kQ$, the matrix
$\displaystyle
\left(
\frac{\partial^2L} {\partial v^i_\alpha \partial v^j_\beta}
\right)_%
{ \genfrac{}{}{0pt}{}{1 \leq i \leq n}{1 \leq j \leq n} }
\!\!\!\!\!(p) $
is invertible.
Otherwise, $L$ is a {\bf singular Lagrangian}.
\end{definition}

\begin{proposition}
A Lagrangian function $L\colon \R^k\times T^1_kQ\to\R$ is regular if and only if $(\d x^\alpha,\omega^\alpha_L,V)$ is a $k$-cosymplectic structure on $\R^k\times T^1_kQ$.
\end{proposition}
\begin{proof}
Taking coordinates, it is easy to see that the forms $\omega_L^\alpha$ are nondegenerate if and only if the matrix
$$
\left(
\frac{\partial^2L} {\partial v^i_\alpha \partial v^j_\beta}
\right)_%
{ \genfrac{}{}{0pt}{}{1 \leq i \leq n}{1 \leq j \leq n} }
\!\!\!\!\!(p)
$$
is invertible for every $1\leq\alpha, \beta\leq k$ and every $p\in\R^k\times T^1_kQ$.
\end{proof}

\begin{definition}
We say that a $k$-vector field $\calX$ of $\R^k\times T^1_kQ$ is a
\textbf{$k$-cosymplectic Lagrangian $k$-vector field}
if it is a solution to equations
\begin{equation}
\label{eq-E-L-k-cosymplectic}
\begin{cases}
i_{X_\alpha}\omega^\alpha_L =
\d E_L +\displaystyle \frac{\partial L}{\partial x^\alpha}\d x^\alpha,
\\
i_{X_\beta}\d x^\alpha =
\delta^\alpha_\beta,
\end{cases}
\end{equation}
where $E_L = \Delta(L)-L$.
We denote by $\X^k_L(\R^k\times T^1_kQ)$ the set of all $k$-cosymplectic Lagrangian $k$-vector fields.
\end{definition}
Equations \eqref{eq-E-L-k-cosymplectic} are called \textbf{$k$-cosymplectic Lagrangian equations}.
			
Notice that if $L$ is regular, then $(\d x^\alpha,\omega^\alpha_L,V)$ is a $k$-cosymplectic structure on $\R^k\times T^1_kQ$. We denote by $\Reeb_\alpha^L$ the corresponding Reeb vector fields.
Hence, if we write the $k$-cosymplectic Hamilton equations for the system
$(\R^k\times T^1_kQ,\d x^\alpha,\omega^\alpha_L,L)$
we get
\begin{equation}\label{eq-E-L-k-cosymplectic-2}
\begin{cases}
i_{X_\alpha}\omega^\alpha_L = \d E_L - \Reeb_\alpha^L(E_L)\d x^\alpha, \\
i_{X_\beta}\d x^\alpha = \delta^\alpha_\beta,
\end{cases}
\end{equation}
which are equivalent to \eqref{eq-E-L-k-cosymplectic}.

If $\mathcal{X}$ is an integrable {\sc sopde} which is a solution to (\ref{eq-E-L-k-cosymplectic-2}),
then its integral sections are solutions to the Euler-Lagrange equations for $L$
$$
\displaystyle\sum_{\alpha=1}^k \left(\displaystyle\frac{\partial^2L}{\partial x^\alpha \partial v^i_\alpha} \, + \,\displaystyle\frac{\partial^2L}{\partial q^j \partial v^i_\alpha} \,\displaystyle\frac{\partial \psi^j}{\partial x^\alpha} \, +
\,\displaystyle\frac{\partial^2L}{\partial v^j_\beta \partial v^i_\alpha} \,
\displaystyle\frac{\partial^2 \psi^j}{\partial x^\alpha \partial x^\beta}\,\right) = \,
\displaystyle\frac{\partial L}{ \partial q^i} \ .
$$

We saw in the Hamiltonian framework that the Reeb vector fields appear in the equations but do not appear in the solutions.
In the Lagrangian framework one can go a step further
and write the system of equations \eqref{eq-E-L-k-cosymplectic-2}
without the Reeb vector fields
\cite{BBLSV-2015}.
Consider the following Poincaré--Cartan 1-forms:
\begin{equation*}
    \Theta^\alpha_L =
    \theta^\alpha_L + \big( \delta^\alpha_\beta L - \Delta^\alpha_\beta(L) \big) \,\d t^\beta \ ,
\end{equation*}
where $\displaystyle
\Delta^\alpha_\beta =v^i_\beta\frac{\partial}{\,\partial v_\alpha^i}$.
Defining $\Omega^\alpha_L = -\d \Theta^\alpha_L$,
the system of equations \eqref{eq-E-L-k-cosymplectic-2} can be written as
\begin{equation*}
\begin{cases}
i_{X_\alpha}\Omega^\alpha_L = (k-1) \,\d L\,, \\
i_{X_\beta}\d x^\alpha = \delta^\alpha_\beta\,,
\end{cases}
\end{equation*}
which is equivalent to~\eqref{eq-E-L-k-cosymplectic}.
To our knowledge, this rewriting without Reeb vector fields cannot be done in the Hamiltonian framework.


\subsection{The Legendre map}
\label{Legmap}

Given a Lagrangian $L\colon\R^k \times T^1_kQ \to \R$ the  Legendre map
$\mathcal{F}L\colon \R^k\times T^1_kQ \longrightarrow \R^k\times (T^1_k)^{\;*}Q$ is
defined as follows:
$$
\mathcal{F}L(t,{v_1}_q,\ldots , {v_k}_q)=(t,\ldots, [\mathcal{F}L(t,{v_1}_q,\ldots ,
{v_k}_q)]^\alpha,\ldots )\ ,
$$ where
$$ [\mathcal{F}L(t,{v_1}_q,\ldots , {v_k}_q)]_q^\alpha(u_q)=
\displaystyle\frac{d}{ds}\Big\vert_{s=0}\displaystyle L\left(
t,{v_1}_q, \dots ,{v_\alpha}_q+su_q, \ldots , {v_k}_q \right)
\quad ; \quad u_q\in T_qQ \ .
$$
It is locally given by
$$
\mathcal{F}L:(t^\alpha,q^i,v^i_\alpha)  \longrightarrow  \left(t^\alpha,q^i,
\frac{\displaystyle\partial L}{\displaystyle\partial v^i_\alpha }
\right)\, .
$$
Let us observe that the Lagrangian forms can also be defined as
$\theta_L^\alpha=\mathcal{F}L^{\;*}\theta^\alpha$ and
$\omega_L^\alpha=\mathcal{F}L^{\;*}\omega^\alpha$.

The Lagrangian function $L$ is regular if, and only if, the corresponding
Legendre map $\mathcal{F}L$ is a local diffeomorphism.
In the particular case that $\mathcal{F}L$ is a global diffeomorphism, $L$ is said to be a
\textbf{hyperregular Lagrangian}.

A singular Lagrangian $L$ is called \textbf{almost-regular} if
$\mathcal{P}:= \mathcal{F}L(T^1_kQ)$ is
a closed submanifold of $\R^k\times(T^1_k)^{\;*}Q$,
the Legendre map $\mathcal{F}L$ is a submersion onto its image, and
the fibres $\mathcal{F}L^{-1}(\mathcal{F}L(v))$,
for every $v\in \R^k\times T^1_kQ$, are
connected submanifolds of $\R^k\times T^1_kQ$.
In this last case, there exists a Hamiltonian formalism associated with the original Lagrangian system,
which is developed on the submanifold $\mathcal{P}$.

{\bf Remark:} {\rm
If $L$ is regular,
$(dt^\alpha,\omega_L^\alpha,V)$
is a $k$-cosymplectic structure on 
${\R^k\times T^1_kQ}$,
where
$V = \ker \,
T \pi_{\R^k\times Q} =
\left\langle
\frac{\displaystyle\partial}{\displaystyle\partial v^i_\alpha}
\right\rangle$
is the vertical distribution of
the vector bundle
$\pi_{\R^k\times Q}\colon\R^k\times T^1_kQ \to \R^k\times Q$.
If $L$ is almost-regular,
then $\mathcal{P}$ is a $k$-precosymplectic manifold
with the  $k$-precosymplectic structure inherited from the above one.}

\section{Constraint algorithm for \texorpdfstring{$k$}--precosymplectic field theories}

In this section we generalize the algorithm developed in
\cite{gra09} for $k$-presymplectic field theories to the
case of $k$-precosymplectic field theories.

We are going to consider $k$-precosymplectic manifolds $M=\R^k\times P$ where $P$ is a $k$-presymplectic manifold. These manifolds have Darboux coordinates and
we also have uniquely determined a collection of Reeb vector fiels $\Reeb_1,\dotsc,\Reeb_k$.
These cases are the most usual appearing in the applications
of the $k$-symplectic and $k$-cosymplectic formulation in classical field theories and applied mathematics.

In the same way as for $k$-cosymplectic field theories, we define:
			
\begin{definition}
A {\bf $k$-precosymplectic Hamiltonian system} is given by a family
$(M,\omega^\alpha,$ $\eta^\alpha,V,\gamma)$,
where $(M,\omega^\alpha,\eta^\alpha,V)$ is a $k$-precosymplectic manifold where $M = \R^k\times P$ and $P$ is a $k$-presymplectic manifold and $\gamma\in\Omega^1(M)$ is a closed $1$-form called the \textbf{Hamiltonian $1$-form}. Since $\gamma$ is closed, by Poincaré's Lemma, $\gamma=\d h$ for some $h\in C^\infty(U)$, $U\subset M$, which is called a {\bf local Hamiltonian function}.

A $k$-vector field $\calX = (X_1,\dotsc,X_k)\in\X^k(M)$ is said to be a \textbf{$k$-precosymplectic Hamiltonian $k$-vector field} if it is a solution to the system of equations
\begin{equation*}
	\begin{cases}
	i_{X_\alpha}\omega^\alpha = \gamma - \gamma(\Reeb_\alpha)\eta^\alpha,\\
	i_{X_\alpha}\eta^\beta = \delta^\beta_\alpha,
	\end{cases}
\end{equation*}
\end{definition}
			
The solutions to the field equations defined by the $k$-presymplectic Hamiltonian system
$(M,\omega^\alpha,\eta^\alpha,V,\gamma)$
are the integral sections of these
$k$-precosymplectic Hamiltonian $k$-vector fields.

{\bf Remark:} 
{\rm
Notice that in the case $k=1$, we recover the case of singular non-autonomous mechanics studied in
\cite{chi94}.
In that case, the Poincar\'e--Cartan $2$-form is widely used in the development of the constraint algorithm.}

We want an algorithm that allows us to find a submanifold $N\hookrightarrow M$ where the system of equations \eqref{eq-Hamilton-k-cosymplectic} has solutions tangent to $N$.
In order to find this submanifold $N$ (if it exists!) we develop an algorithm which introduces some constraints in every step that provides us a sequence of submanifolds
\begin{equation*}
\dotsb \hookrightarrow M_j \hookrightarrow \dotsb \hookrightarrow M_2 \hookrightarrow M_1 \hookrightarrow M
\end{equation*}
which in favorable cases will end in the {\bf final constraint submanifold} $N$. Notice that this manifold may be a union of isolated points ($\dim N = 0$) or be empty. These cases have no interest for us, we are only interested in cases where we have a final constraint submanifold of dimension greater than 0.
			
\begin{theorem}   
\label{thm-k-co-algorithm}
Consider a $k$-precosymplectic Hamiltonian system $(M,\omega^\alpha,\eta^\alpha,V,\gamma)$, a submanifold $C\hookrightarrow M$
and a $k$-vector field $\calX\colon C\to (T^1_k)_CM$
such that $\calX_p\in (T_k^1)_pC$ for every $p\in C$.
Under this setting, the following two conditions are equivalent:
\begin{enumerate}[(1)]
\item
There exists a $k$-vector field $\calX = (X_\alpha)\colon C\to (T^1_k)_CM$ tangent to $C$ such that the system of equations
\begin{equation}
\label{eq-thm-k-co}
	\begin{cases}
	i_{X_\alpha}\omega^\alpha = \gamma - \gamma(\Reeb_\alpha) \eta^\alpha,\\
	i_{X_\alpha}\eta^\beta = \delta^\beta_\alpha,
\end{cases}
\end{equation}
holds on $C$.
\item
For every $p\in C$, there exists
$\calZ_p = (Z_\alpha)_p\in (T^1_k)_pC$
such that, if
$\widetilde\gamma_p = \gamma_p({\Reeb_\alpha}_p)\eta^\alpha_p$,
then
$$
i_{{Z_\alpha}_p}\eta^\beta_p =
\delta_\alpha^\beta \quad , \quad \sum_\alpha\eta^\alpha_p + \widetilde\gamma_p =
\flat(\calZ_p) \ .
$$
\end{enumerate}
\end{theorem}   
\begin{proof}
Consider the $k$-vector $\calZ_p = \calX_p\in (T^1_k)_pC$.
It is clear that $i_{{Z_\alpha}_p}\eta^\beta_p = \delta^\beta_\alpha$ for every $p\in C$ and that
\begin{equation*}
	\flat(\calZ_p) = i_{{Z_\alpha}_p}\omega^\alpha_p + (i_{{Z_\alpha}_p}\eta^\alpha_p)\eta^\alpha_p = \widetilde\gamma_p + \sum_\alpha\eta^\alpha_p.
\end{equation*}
				
Conversely, let us suppose that for every $p\in C$, there exists $\calZ_p\in (T^1_k)_pC$ such that
$i_{{Z_\alpha}_p}\eta^\beta_p =
\delta^\beta_\alpha$ and $\flat(\calZ_p) =
\widetilde\gamma_p + \sum_\alpha\eta^\alpha_p$.
Let $p\in C$.
We consider a Darboux chart
$(\U, \{x^\alpha,y^i,y^\alpha_{i_\alpha};z^j\})$
around $p$ and hence,
\begin{gather*}
\eta^\alpha = \d x^\alpha,
\\
\omega^\alpha = \sum_{i\in I_\alpha}\d y^i\wedge\d y_i^\alpha,
\\
\gamma =
\frac{\partial h}{\partial y^i}\d y^i +
\frac{\partial h}{\partial y_{i_\alpha}^\alpha}\d y_{i_\alpha}^\alpha +
\frac{\partial h}{\partial x^\alpha}\d x^\alpha +
\frac{\partial h}{\partial z^j}\d z^j.
\end{gather*}
In these Darboux coordinates,
$\widetilde\gamma = \gamma - \gamma(\Reeb_\alpha)\eta^\alpha$ is
\begin{equation*}
	\widetilde\gamma =
	\frac{\partial h}{\partial y^i}\d q^i + \frac{\partial h}{\partial y_{i_\alpha}^\alpha}\d y_{i_\alpha}^\alpha + \frac{\partial h}{\partial z^j}\d z^j.
\end{equation*}
From now on, we will omit the point $p$ everywhere in order to simplify the notation.
We write our $k$-vector $\calZ$ in coordinates:
\begin{equation*}
Z_\alpha =
A^\beta_\alpha\frac{\partial}{\partial x^\beta} +
B^i_\alpha\frac{\partial}{\partial y^i} +
C^\beta_{\alpha,i_\beta}\frac{\partial}{\partial y_{i_\beta}^\beta} +
D_\alpha^j\frac{\partial}{\partial z^j}.
\end{equation*}
Now let us compute its image by the morphism $\flat$:
\begin{align*}
\flat(\calZ) &= \sum_\alpha i_{Z_\alpha}\omega^\alpha + (i_{Z_\alpha}\eta^\alpha)\eta^\alpha \\
&= \sum_\alpha\sum_{i\in I_\alpha} i_{Z_\alpha}(\d y^i\wedge\d y_i^\alpha) + \sum_\alpha (i_{Z_\alpha}\d x^\alpha)\d x^\alpha \\
&= \sum_\alpha\sum_{i\in I_\alpha} (i_{Z_\alpha}\d y^i)\,\d y_i^\alpha - \sum_\alpha\sum_{i\in I_\alpha} \d y^i\,(i_{Z_\alpha}\d y_\alpha^i) + \sum_\alpha (i_{Z_\alpha}\d x^\alpha)\d x^\alpha \\
&= \sum_\alpha\sum_{i\in I_\alpha} B_\alpha^i\d y^\alpha_i - \sum_\alpha\sum_{i\in I_\alpha}C^\alpha_i\d y^i + \sum_\alpha A^\alpha_\alpha\d x^\alpha.
\end{align*}
Comparing this expression with
\begin{equation*}
	\sum_\alpha\eta^\alpha + \widetilde\gamma = \sum_\alpha\d x^\alpha + \frac{\partial h}{\partial y^i}\d y^i + \frac{\partial h}{\partial y_{i_\alpha}^\alpha}\d y_{i_\alpha}^\alpha + \frac{\partial h}{\partial z^j}\d z^j,
\end{equation*}
we get the following conditions on $Z$:
\begin{equation*}
	A^\alpha_\alpha = 1,\qquad \frac{\partial h}{\partial z^j} = 0,\qquad \frac{\partial h}{\partial y^i} = -\sum_{\substack{\alpha\mbox{ \scriptsize such} \\ \mbox{\scriptsize that } i\in I_\alpha}}C^\alpha_{\alpha,i},\qquad \frac{\partial h}{\partial y_{i_\alpha}^\alpha} = B^{i_\alpha}_\alpha.
\end{equation*}
Furthermore, we know by hypothesis that $A^\beta_\alpha = \delta^\beta_\alpha$.
The second condition $\displaystyle \frac{\partial h}{\partial z^j} = 0$ is a compatibility condition of the Hamilton equations in the $k$-precosymplectic case.
It can be stated in the following way: the Hamiltonian function cannot depend on the so-called {\sl gauge variables} $z^j$. The third and fourth equations, along with the condition $A_\alpha^\beta = \delta_\alpha^\beta$, are equivalent to the system of equations \eqref{eq-thm-k-co} when written in coordinates (see equation \eqref{eq-k-co-coordinates}). This concludes the proof.
\end{proof}

Using the previous theorem we can give a description of the constraint algorithm.
First of all, we must restrict ourselves to the points such that
$\displaystyle\gamma\Big(\frac{\partial}{\,\partial z^j}\Big)=0$, $\forall j$,
because it is a compatibility condition of the system.
The $j$-ary constraint submanifold $M_j\subset M_{j-1}$ is defined as
{\small 
\begin{equation*}
	M_j = \Big\{p\in M_{j-1}\,\vert\,\exists \calZ =
	(Z_\alpha)\in (T^1_k)M_{j-1}
	\mbox{ such that } \flat(\calZ) = \widetilde\gamma + \sum_\alpha\eta^\alpha
	\mbox{ and }
	i_{Z_\alpha}\eta^\beta = \delta^\beta_\alpha\Big\},
\end{equation*}}%
where $M_0 = M$.
			
\begin{definition}
Let $C\hookrightarrow M$ be a submanifold of a $k$-precosymplectic manifold~$M$.
The \textbf{$k$-precosymplectic orthogonal complement} of $C$ is
the annihilator
\begin{equation*}
	TC^\perp = \left(\flat\left((T^1_k)C\cap D_C\right)\right)^0
\end{equation*}
where $D_C$ is the set of all $k$-vectors
$\calZ_p=(Z_\alpha)_p$ on $C$
such that
$i_{{Z_\alpha}_p}\eta_p^\beta = \delta^\beta_\alpha$.
\end{definition}

With this definition and Theorem \ref{thm-k-co-algorithm}
we can give an alternative characterization of the constraints submanifolds:
\begin{equation*}
M_j = \Big\{p\in M_{j-1}\,\vert\,\widetilde\gamma_p+\sum_\alpha\eta_p^\alpha\in ((TM_{j-1})_p^\perp)^0\Big\} .
\end{equation*}
Although this allows us to effectively compute the constraints at every step of the algorithm,
an alternative and equivalent way to compute the constraint submanifolds given by the
\textbf{$k$-precosymplectic constraint algorithm},
which is much more operational, is the following:
\begin{enumerate}[(i)]
\item
Obtain a local basis $\{Z_1,\dotsc,Z_r\}$ of $(TM)^\perp$.
\item
Use Theorem \ref{thm-k-co-algorithm} to obtain a set of independent constraint functions
\begin{equation}
f_\mu = i_{Z_\mu}(\widetilde\gamma + \sum_\alpha\eta^\alpha) \ ,
\label{constra}
\end{equation}
which define the submanifold $M_1\hookrightarrow M$.
\item
Compute solutions $\calX = (X_\alpha)$ to equations \eqref{eq-Hamilton-k-cosymplectic}.
\item
Impose the tangency condition of $X_1,\dotsc,X_k$ on $M_1$.
\item
Iterate item (iv) until no new constraints appear.
\end{enumerate}
		
If this iterative procedure ends in a submanifold $M_l$ with nonzero dimension, then we can ensure the existence of global solutions to equations \eqref{eq-Hamilton-k-cosymplectic} on this submanifold~$M_l$.

{\bf Remark:} 
\rm
The constraint algorithm works, in particular,
for a singular Lagrangian field theory
$(\R^k\times T^1_kQ,\omega^\alpha_L,\d x^\alpha,L)$,
and for its associated Hamiltonian formalism on $\mathcal{P}$.
Nevertheless, in the case of the Lagrangian formalism,
the problem of finding {\sc sopde} multivector fields
which are solutions to the field equations is not considered in this algorithm.
These {\sc sopde} multivector fields can be obtained, in some cases,
by fixing some arbitrary functions in the general solution to the field equations
on the final constraint submanifold $M_f$.
However, in general,
looking for {\sc sopde} multivector fields solution
leads to new constraints
which define a new submanifold $M_{f^\prime} \hookrightarrow M_f$;
hence, the tangency condition may originate more constraints and,
in the best of cases,
we obtain a new final constraint submanifold
$S_{f^\prime} \hookrightarrow M_{f^\prime}$
where there are {\sc sopde} multivector fields solutions
tangent to $S_{f^\prime}$.
In the examples analyzed in Section \ref{afflag}
we give some insights on how to proceed in these cases
(see, for instance,
\cite{dLe02}
for a deep study of these topics in singular mechanics).
Nevertheless a rigorous intrinsic characterization
of all of these additional ``{\sc sopde} constraints'' in field theories
is still an open problem.

Finally, notice that we can treat $k$-presymplectic field theories as a particular case of $k$-precosymplectic field theories.
In this case, we do not have the $1$-forms $\eta^\alpha$ and we recover the $k$-presymplectic algorithm described in
\cite{gra09}.

\section{Examples}

\subsection{Affine Lagrangians}
\label{afflag}

In classical field theory affine Lagrangians are used to describe
some relevant models in Physics such as, for instance,
the so-called Einstein--Palatini (or \textit{metric-affine}) approach
to gravitation,
and Dirac fermion fields
\cite{GMS-1997},
among others.

Let $Q$  be the configuration manifold of a field theory.
The bundle 
$\bar\tau_1\colon\R^k\times T^1_kQ\to\R^k$
represents its non-autonomous phase space
of $k$-velocities,
and has coordinates $(x^\alpha, q^i, v^i_\alpha)$. 
We consider an affine Lagrangian
$L\colon\R^k\times T^1_kQ\to\R$ of the form
\begin{equation}
\label{eq-afflag}
    L(x^\alpha, q^i, v^i_\alpha) = 
    f^\mu_j(q^i) v^j_{\mu}+g(x^\alpha, q^i)  \; .
\end{equation}
Such a function is the sum of
the pullbacks to $\R^k\times T^1_kQ$
of two functions:

- a linear function 
$T^1_kQ \to \R$ on the fibers of the bundle
$T^1_kQ\to Q$;

- an arbitrary function 
$\R^k \times Q \to \R$.

\subsubsection*{Lagrangian formalism}

Associated with this affine Lagrangian 
we have
\beann
E_L&=&\Delta(L)-L=-g(x^\alpha, q^i)\in C^\infty(\R^k\times T^1_kQ) \ , \\
\omega_L^\alpha&=&
-\derpar{f^\alpha_k}{q^j}\,\d q^j\wedge\d q^k
\in{\mit\Omega}^2(\R^k\times T^1_kQ) \ ,
\eeann
which define a $k$-precosymplectic structure $(\omega_L^\alpha, \diff x^\alpha, \mathcal{V})$ in $\R^k\times T^1_kQ$, where
${\cal V}=\left\langle\displaystyle\derpar{}{v^i_\mu}\right\rangle$
and we can take $\displaystyle\Reeb_\alpha= \frac{\partial}{\partial x^\alpha}$
as the Reeb vector fields.
For a $k$-vector field $\calX = (X_1,\dotsc,X_k)\in\X^k(\R^k\times T^1_kQ)$
with
\beq
X_\alpha=
\derpar{}{x^\alpha}+F^l_\alpha\,\derpar{}{q^l}+G^l_{\alpha\nu}\,\derpar{}{v^l_\nu}
\in{\mathfrak X}(\R^k\times T^1_kQ) \ ,
\label{semhol}
\eeq
the Euler--Lagrange equation \
$\displaystyle i_{X_\alpha}\omega^\alpha_L=\d E_L+\derpar{L}{x^\mu}\,\d x^\mu$ \
gives
\begin{equation}
F^l_\alpha\left(\derpar{f^\alpha_l}{q^j}-\derpar{f^\alpha_j}{q^l} \right)\d q^j=-\derpar{g}{q^j}\,\d q^j
\quad \Longleftrightarrow \quad
\derpar{g}{q^j}+F^l_\alpha\left(\derpar{f^\alpha_l}{q^j}-\derpar{f^\alpha_j}{q^l}\right)=0\ .
\label{2eq}
\end{equation}
This is a system of (linear) equations for the component functions $F^l_\alpha$,
which allows us to determine (partially) these functions and, eventually, gives raise to
constraints functions (depending on the rank of the matrices involved).
If this last situation happens, then the constraint algorithm follows
by demanding the tangency condition for the vector fields $X_\alpha$.
Observe also that, in any case, in these vector fields, the coefficients 
$G^i_{\alpha\nu}$ are undetermined.

If we look for semi-holonomic $k$-vector fields ${\cal X}$,
it implies that $F^k_\nu=v^k_\nu$ in \eqref{semhol}.
Then, equations \eqref{2eq} read
$$
\derpar{g}{q^j}+v^l_\alpha\left(\derpar{f^\alpha_l}{q^j}-\derpar{f^\alpha_j}{q^l} \right)=0
$$
which are constraints.
Then the tangency condition for the vector fields 
$$
X_\nu=\derpar{}{x^\nu}+v^l_\nu\,\derpar{}{q^l}+G^l_{\nu\alpha}\,\derpar{}{v^l_\alpha} \ ,
$$
leads to
$$
\derpar{g}{q^j}+G^l_{\nu\alpha}\left(\derpar{f^\alpha_l}{q^j}-\derpar{f^\alpha_j}{q^l} \right)=0
$$
which allows us to determine (partially) the functions $G^k_{\nu\alpha}$ and, eventually, gives raise to
constraints functions, depending on the rank of the matrix 
$\displaystyle\left(\derpar{f^\alpha_l}{q^j}-\derpar{f^\alpha_j}{q^l} \right)$.
In this last case, the constraint algorithm continues
by demanding again the tangency condition. 

\subsubsection*{Hamiltonian formalism}

The bundle $\bar\pi_1\colon\R^k\times (T^1_k)^\ast Q\to\R^k$ is the 
phase space of $k$-momenta. The Legendre map 
${\cal F}L\colon\R^k\times T^1_k Q\to\R^k\times (T^1_k)^\ast Q$
associated to the Lagrangian $L$ is
$$
x^\mu\circ{\cal F}L=x^\mu\quad ,\quad
q^i\circ{\cal F}L=q^i\quad ,\quad
p_i^\mu\circ{\cal F}L=\derpar{L}{v^i_\mu}=f^\mu_i(q^j)\ .
$$
Observe that ${\cal P}={\cal F}L(\R^k\times T^1_k Q)$ is defined 
by the constraints $p_i^\mu=f^\mu_i(q^j)$; hence
it is the image of a section $\xi\colon\R^k\times Q\to\R^k\times (T^1_k)^\ast Q$
of the projection $(\pi_Q)_{(1,0)}\colon\R^k\times (T^1_k)^*Q\to\R^k\times Q$, 
and then it can be identified in a natural way with $\R^k\times Q$.
Therefore, as $\xi\circ\tau_1$ is a surjective submersion with connected fibres, then so is
${\cal F}L_0\colon \R^k\times T^1_k Q\to{\cal P}$ 
(the restriction of ${\cal F}L$ onto its image ${\cal P}$),
since ${\cal F}L_0=\xi\circ\tau_1$. 
In conclusion, affine Lagrangians are almost regular Lagrangians
and thus ${\cal P}$ is an embedded submanifold of $\R^k\times (T^1_k)^\ast Q$, 
which is diffeomorphic to $\R^k\times Q$.

Therefore we can introduce
\beann
h&=&-g(x^\alpha, q^i)\in C^\infty({\cal P}) \ , \\
\omega^\alpha&=&
-\derpar{f^\alpha_k}{q^j}\,\d q^j\wedge\d q^k
\in{\mit\Omega}^2({\cal P}) \ ,
\eeann
such that ${\cal F}L_0^*E_L=h$ and 
${\cal F}L_0^*\omega^\alpha_L=\omega^\alpha$.
As above, $\eta^\alpha=\d x^\alpha$ and 
the Reeb vector fields are $\Reeb_\alpha=\displaystyle \frac{\partial}{\partial x^\alpha}$.
Then, for a $k$-vector field $\calX = (X_1,\dotsc,X_k)\in\X^k({\cal P})$,
with
$$
X_\alpha=\derpar{}{x^\alpha}+F^l_\alpha\,\derpar{}{q^l}\in{\mathfrak X}({\cal P}) \ ,
$$
the Hamilton equation \
$\displaystyle i_{X_\alpha}\omega^\alpha=\d h-\Reeb_\alpha(h)\,\d x^\mu$ \
leads to
\beq
-F_\alpha^l\left(\frac{\partial f_j^\alpha}{\partial q^l} - \frac{\partial f_l^\alpha}{\partial q^j}\right)\d q^j=
-\derpar{g}{q^j}\,\d q^j
\quad \Longleftrightarrow \quad
\derpar{g}{q^j}+
F^l_\alpha\left(\derpar{f^\alpha_l}{q^j}-\derpar{f^\alpha_j}{q^l} \right)=0
\ ,
\label{22}
\eeq
As in the Lagrangian formalism, this system of (linear) equations
allows us to determine (partially) the component functions $F^l_\alpha$
and, eventually, could give constraints functions (depending on the rank of the matrices involved).
If this last situation happens, then the constraint algorithm follows
by demanding the tangency condition for the vector fields $X_\alpha$.

\subsection{A simple affine Lagrangian model}
\label{2ndex}


\subsubsection*{Lagrangian formalism}

The configuration manifold is $\R^2\times Q =\R^2\times\R^2$,
 with coordinates $(x^1,x^2;q^1,q^2)$.
The Lagrangian formalism takes place in
$\R^2 \times \oplus^2 T Q$,
with coordinates
$(x^1,x^2,q^1,q^2,v^1_1,v^1_2,v^2_1,v^2_2)$,
and the Lagrangian is
$$
L = q^2v_1^1 - q^1 v_2^2 + q^1q^2x^1 \,;
$$
that is,
the functions in equation~\eqref{eq-afflag} are
$$
f^1_1=q^2 \ , \quad f^2_1=0 \ , \quad
f^1_2=0\ , \quad f^2_2=-q^1\ , \quad
g=q^1q^2x^1 \,.
$$
We have the forms
$$
\eta^1 = \d x^1 \:,\quad 
\eta^2 = \d x^2 \:;\quad
\omega^1_L= \d q^1\wedge\d q^2 \:,\quad
\omega^2_L= \d q^1\wedge\d q^2 \:,
$$
and the Reeb vector fields \ 
$\displaystyle
\Reeb_1^L = \frac{\partial}{\partial x^1}\,, \
\Reeb_2^L = \frac{\partial}{\partial x^2}.$ \
The energy is simply
$$
E_L= - q^1 q^2x^1 \ ,
$$
and, if $\calX = (X_1,X_2) \in \X^2(\R^2 \times \oplus^2 T Q)$
is a generic $2$-vector field 
with
\begin{gather*}
X_1 =  \frac{\partial}{\partial x^1} + 
  F^1_1\frac{\partial}{\partial q^1} + F^2_1\frac{\partial}{\partial q^2} + 
  G^1_{11}\frac{\partial}{\partial v_1^1}+G^1_{12}\frac{\partial}{\partial v_2^1}+
  G^2_{11}\frac{\partial}{\partial v_1^2}+G^2_{12}\frac{\partial}{\partial v_2^2} \ , 
\\
X_2 =\frac{\partial}{\partial x^2} + 
  F^1_2\frac{\partial}{\partial q^1} + F^2_2\frac{\partial}{\partial q^2} + 
  G^1_{21}\frac{\partial}{\partial v_1^1}+G^1_{22}\frac{\partial}{\partial v_2^1}+
  G^2_{21}\frac{\partial}{\partial v_1^2}+G^2_{22}\frac{\partial}{\partial v_2^2} \ ;
\end{gather*}
then the Lagrangian equation \
$i_{X_\alpha}\omega^\alpha_L=\d E_L-\Reeb_\alpha^L(E_L)\,\d x^\alpha$ \
is
$$
F_1^1\d q^2 - F_1^2\d q^1 + F_2^1\d q^2 - F_2^2\d q^1=
-q^2x^1\d q^1-q^1x^1\d q^2\ ,
$$
and conditions \eqref{2eq} read
$$
F_1^2 + F_2^2 = q^2x^1 \:, \quad
F_1^1 + F_2^1 = -q^1x^1 \:,
$$
which can also be written as
\begin{equation*}
    \begin{pmatrix}
        0 & 1 & 0 & 1\\
        -1 & 0 & -1 & 0
    \end{pmatrix}
    \begin{pmatrix}
        F_1^1\\
        F_1^2\\
        F_2^1\\
        F_2^2
    \end{pmatrix}
    = \begin{pmatrix}
        q^2x^1\\
        q^1x^1
    \end{pmatrix}\ .
\end{equation*}
Imposing the second order condition, $F_\mu^l = v_\mu^l$, we have that the 2-vector field $\mathcal{X} = (X_1,X_2)$ is
\begin{gather*}
    X_1 = \frac{\partial}{\partial x^1} + v_1^1\frac{\partial}{\partial q^1} + v_1^2\frac{\partial}{\partial q^2} + G_{1\nu}^l\frac{\partial}{\partial v_\nu^l}\\
    X_2 = \frac{\partial}{\partial x^2} + v_2^1\frac{\partial}{\partial q^1} + v_2^2\frac{\partial}{\partial q^2} + G_{2\nu}^l\frac{\partial}{\partial v_\nu^l}
\end{gather*}
and we get the two constraints
$$
\begin{cases}
    \zeta_1 = v_1^2 + v_2^2 - q^2x^1 = 0\\
    \zeta_2 = v_1^1 + v_2^1 + q^1x^1 = 0
\end{cases}
$$
The constraints
$\zeta_1$ and $\zeta_2$ 
define the submanifold 
${\cal S}_1\hookrightarrow\R^2 \times \oplus^2 T Q$.
Next, the tangency conditions on this submanifold lead to
\begin{equation*}
    \begin{cases}
        X_1(\zeta_1) = -q^2 + G_{11}^2 + G_{12}^2 - x^1v_1^2 = 0\\
        X_1(\zeta_2) = q^1 + x^1v_1^1 + G_{11}^1 + G_{12}^1 = 0\\
        X_2(\zeta_1) = -x^1v_2^2 + G_{21}^2 + G_{22}^2 = 0\\
        X_2(\zeta_2) = x^1v_2^1 + G_{21}^1 + G_{22}^1 = 0
    \end{cases}
\end{equation*}
which can be written in matrix form as
\begin{equation*}
    \begin{pmatrix}
        0 & 1 & 0 & 1 & 0 & 0 & 0 & 0\\
        1 & 0 & 1 & 0 & 0 & 0 & 0 & 0\\
        0 & 0 & 0 & 0 & 0 & 1 & 0 & 1\\
        0 & 0 & 0 & 0 & 1 & 0 & 1 & 0\\
    \end{pmatrix}
    \begin{pmatrix}
        G_{11}^1\\
        G_{11}^2\\
        G_{12}^1\\
        G_{12}^2\\
        G_{21}^1\\
        G_{21}^2\\
        G_{22}^1\\
        G_{22}^2\\
    \end{pmatrix}
    = \begin{pmatrix}
        q^2 + x^1v_1^2\\
        -q^1 - x^1v_1^1\\
        x^1v_2^2\\
        -x^1v_2^1
    \end{pmatrix}
\end{equation*}
which allow us to partially determine the coefficients $G_{\alpha\nu}^l$.
Notice that no new constraints appear. Thus, the final constraint submanifold is $\mathcal{S}_1$.

\subsubsection*{Hamiltonian formalism}

The Hamiltonian formalism takes place
in the bundle $\R^2 \times \oplus^2 T^* Q$,
which has coordinates
$(x^1,x^2,y^1,y^2,p^1_1,p^2_1,p^1_2,p^2_2)$.
The Legendre map
$\mathcal{F}L \colon \R^2 \times \oplus^2 T Q \to \R^2 \times \oplus^2 T^* Q$
is given by
$$
(x^1,x^2,y^1,y^2,p^1_1,p^2_1,p^1_2,p^2_2)=
\mathcal{F}L(x^1,x^2,q^1,q^2;v^1_1,v^1_2,v^2_1,v^2_2) =
(x^1,x^2,q^1,q^2;q^2,0,0,-q^1) \:.
$$
Its image is the submanifold ${\cal P}$
of $\R^2 \times \oplus^2 T^* Q$ given by the primary constraints
$$
p^1_1=q^2 \:, \quad 
p^2_1=0 \:, \quad 
p^1_2=0 \:, \quad 
p^2_2=-q^1 \:;
$$
so, we can describe ${\cal P}$ with coordinates
$(x^1,x^2,q^1,q^2)$.
In ${\cal P}$ we have the forms
$$
\eta^1 = \d x^1 \:,\quad 
\eta^2 = \d x^2 \:;\quad
\omega^1 = \d q^1\wedge\d q^2 \:,\quad
\omega^2 = \d q^1\wedge\d q^2 \:,
$$
and the Reeb vector fields \
$\displaystyle
\Reeb_1 = \frac{\partial}{\partial x^1} \,, \
\Reeb_2 = \frac{\partial}{\partial x^2}$. \
The Hamiltonian function is
$$
h = - q^1 q^2 x^1 \:,
$$
Let  
$\calX = (X_1,X_2) \in \X^2({\cal P})$
be a generic $2$-vector field 
with
$$
X_1 = \frac{\partial}{\partial x^1} +
B^1_1\frac{\partial}{\partial q^1} + 
B^2_1\frac{\partial}{\partial q^2} \:,\quad
X_2 = \frac{\partial}{\partial x^2} + 
B^1_2\frac{\partial}{\partial q^1} +  
B^2_2\frac{\partial}{\partial q^2} \:;
$$
then the Hamiltonian equation \
$\displaystyle i_{X_\alpha}\omega^\alpha=\d h-\Reeb_\alpha(h)\,\d x^\mu$ \
gives
$$
B_1^1\d q^2 - B_1^2\d q^1 + B_2^1\d q^2 - B_2^2\d q^1=
-x^1q^2\d q^1 - x^1 q^1\d q^2\ .
$$
In this case, conditions \eqref{22} read
$$
B_1^2 + B_2^2 = x^1 q^2 \:,\quad
B_1^2 + B_2^1 = -x^1 q^1 \:.
$$
which allow us partially determine the coefficients $B_\alpha^j$. 
Notice that, in this case, no new constraints appear.

\subsection{A singular quadratic Lagrangian}

\subsubsection*{Lagrangian formalism}

As another example we consider the Lagrangian function
\begin{equation*}
L = \frac{1}{2e}q_t^2 + \frac{1}{2}\sigma^2 e - \frac{1}{2}\tau q_s^2 \,,
\end{equation*}
with two independent variables
$(t,s)\in\R^2$
and two dependent variables
$(q,e) \in Q= \R\times\R^+$;
the corresponding natural coordinates of
$\R^2\times\oplus^2TQ$ are written
$(t,s;q,e;q_t,q_s,e_t,e_s)$.
Also $\tau\in\R$ is a constant parameter and
$\sigma = \sigma(t,s) \in \mathcal{C}\infty(\R^2)$
is a given function.
This Lagrangian is very similar to the one introduced in
\cite{gra09}
but letting one of its parameters
to be a given function
in order to illustrate the non-autonomous setting.

First we need to compute several geometric objects:
$$
J^t =
\frac{\partial}{\,\partial q_t} \otimes \d q +
\frac{\partial}{\,\partial e_t} \otimes \d e
\,,\quad
J^s =
\frac{\partial}{\,\partial q_s} \otimes \d q +
\frac{\partial}{\,\partial e_s} \otimes \d e
\,,
$$
$$
\Delta_t =
q_t \frac{\partial}{\,\partial q_t} +
e_t \frac{\partial}{\,\partial e_t}
\,,\quad
\Delta_s =
q_s \frac{\partial}{\,\partial q_s} +
e_s \frac{\partial}{\,\partial e_s}
\,,\quad
\Delta = \Delta_t + \Delta_s
\,.
$$
Now we compute the Poincar\'e--Cartan forms:
$$
\theta_L^t = \d L \circ J^t =
\frac 1e \,q_t \,\d q
\,,\quad
\theta_L^s = \d L \circ J^s =
-\tau \,q_s \,\d q
\,,
$$
$$
\omega_L^t = -\d \theta_L^t =
\frac {q_t}{e^2} \,\d e \wedge \d q -
\frac 1e \,\d q_t \wedge \d q
\,,\quad
\omega_L^s = -\d \theta_L^s =
\tau \,\d q_s \wedge \d q
\,.
$$
We also need the Lagrangian energy:
$$
E_L = \Delta(L) - L =
\frac{1}{2e} q_t^2 - \frac12 \sigma^2 e - \frac12 \tau q_s^2
\,.
$$

Now we consider a generic 2-vector field
$\calX = (X_t,X_s)$ on the phase space
$\R^2 \times \oplus^2 T Q$
and consider the $k$-cosymplectic Euler--Lagrange equations
\eqref{eq-E-L-k-cosymplectic}
for it.
After applying the second group of equations,
$i_{X_\beta} \d x^\alpha =
\delta^\alpha_\beta$,
the form of $\calX$ is given by
\begin{gather*}
X_t =
\frac{\partial}{\partial t} +
B_t^q \frac{\partial}{\partial q} +
B_t^e \frac{\partial}{\partial e} +
C_t^{q_t} \frac{\partial}{\partial q_t} +
C_t^{q_s} \frac{\partial}{\partial q_s} +
C_t^{e_t} \frac{\partial}{\partial e_t} +
C_t^{e_s} \frac{\partial}{\partial e_s}
\,,
\\
X_s =
\frac{\partial}{\partial s} +
B_s^q \frac{\partial}{\partial q} +
B_s^e \frac{\partial}{\partial e} +
C_s^{q_t} \frac{\partial}{\partial q_t} +
C_s^{q_s} \frac{\partial}{\partial q_s} +
C_s^{e_t} \frac{\partial}{\partial e_t} +
C_s^{e_s} \frac{\partial}{\partial e_s}
\,.
\end{gather*}
Now we apply the first equation
of the $k$-cosymplectic Euler--Lagrange equations,
$\displaystyle
i_{X_\alpha} \omega^\alpha_L =
\d E_L + \frac{\partial L}{\partial x^\alpha} \d x^\alpha
$.
Equating the coefficients of the differentials we obtain:
\begin{gather*}
B_t^e = \frac{e^2}{q_t}
\left( \frac 1e C_t^{q_t} - \tau C_s^{q_s} \right)
,\quad
B_t^q = q_t
\,,\quad
B_s^q = q_s
\,,
\\
\frac{q_t^2}{e^2} = \sigma^2
\,.
\end{gather*}
From these equations,
three coefficients of $\calX$
are determined from the variables and the other coefficients;
the last equation is a constraint,
$$
\zeta_1 = \frac12 \left( \frac{q_t^2}{e^2} - \sigma^2 \right)
\,,
$$
whose vanishing defines a submanifold
$\mathcal{S}_1$ of $\R^2 \times \oplus^2 T Q$.
At this stage, $\calX$ has nine undetermined coefficients.

Finally,
we analyze the tangency of $\calX$
(that is, of its components $X_t$, $X_s$)
to~$\mathcal{S}_1$.
Imposing
$X_t(\zeta_1) \vert_{\mathcal{S}_1}= 0$,
$X_s(\zeta_1) \vert_{\mathcal{S}_1}= 0$,
we obtain two additional relations between the undetermined coefficients (on $\mathcal{S}_1$),
and no more constraints.

\medskip

To complete our analysis
we will study
the $k$-cosymplectic Euler--Lagrange equations
with the second-order partial differential equation condition.
The generic expression of $\calX$ is given by
\begin{gather*}
X_t =
\frac{\partial}{\partial t} +
q_t \frac{\partial}{\partial q} +
e_t \frac{\partial}{\partial e} +
C_t^{q_t} \frac{\partial}{\partial q_t} +
C_t^{q_s} \frac{\partial}{\partial q_s} +
C_t^{e_t} \frac{\partial}{\partial e_t} +
C_t^{e_s} \frac{\partial}{\partial e_s}
\,,
\\
X_s =
\frac{\partial}{\partial s} +
q_s \frac{\partial}{\partial q} +
e_s \frac{\partial}{\partial e} +
C_s^{q_t} \frac{\partial}{\partial q_t} +
C_s^{q_s} \frac{\partial}{\partial q_s} +
C_s^{e_t} \frac{\partial}{\partial e_t} +
C_s^{e_s} \frac{\partial}{\partial e_s}
\,.
\end{gather*}
Now the first equation of
\eqref{eq-E-L-k-cosymplectic}
yields two identities,
a relation between the coefficients,
namely
$\displaystyle
e_t \frac{q_t}{e^2} - \frac 1e C_t^{q_t} + \tau C_s^{q_s} = 0
$,
and the same constraint $\zeta_1$ as above.
The tangency of $\calX$ to the submanifold~$\mathcal{S}_1$
determines the coefficients
$C_t^{q_t}$ and $C_s^{q_t}$ (on $\mathcal{S}_1$),
and no more constraints appear.
So $\calX$ is left with 5 undetermined coefficients.

\subsubsection*{Hamiltonian formalism}

The Hamiltonian formalism takes place in
$\R^2\times\oplus^2T^\ast Q$,
where we use natural coordinates
$(t,s;q,e;p^t,p^s,\pi^t,\pi^s)$.

The Legendre map
$\mathcal{F}L \colon \R^2 \times \oplus^2 TQ \to \R^2 \times \oplus^2 T^\ast Q$
is given by
\begin{equation*}
	\mathcal{F}L(t,s;q,e;q_t,q_s,e_t,e_s) =
    \left( t,s; q,e;\frac{1}{e}q_t,-\tau q_s,0,0 \right).
\end{equation*}
The primary Hamiltonian constraint submanifold
$\mathcal{P}$ of $\R^2\times\oplus^2T^\ast Q$
is defined by the constraints
\begin{equation*}
	\pi^t = 0,
    \quad
    \pi^s = 0.
\end{equation*}
Using $(t,s;q,e;p^t,p^s)$ as coordinates on the submanifold $\mathcal{P}$,
its $2$-precosymplectic structure is given by
\begin{equation*}
	\eta^t = \d t,\quad \eta^s = \d s,
    \quad
    \omega^t = \d q\wedge\d p^t,
    \quad
    \omega^s = \d q\wedge\d p^s.
\end{equation*}
Then,
\begin{equation*}
	\ker\eta^t \cap \ker\eta^s \cap \ker\omega^t \cap \ker\omega^s =
    \left\langle \frac{\partial}{\partial e} \right\rangle.
\end{equation*}
The Reeb vector fields are
\begin{equation*}
	\Reeb_t = \frac{\partial}{\partial t},
    \quad
    \Reeb_s = \frac{\partial}{\partial s}.
\end{equation*}
The Hamiltonian function on $\mathcal{P}$ is
\begin{equation*}
	h = \frac{1}{2} \,e \, (p^t)^2 - \frac{1}{2} \sigma^2 e - \frac{1}{2\tau} (p^s)^2.
\end{equation*}
Consider a $2$-vector field $\calX = (X_t,X_s)\in\X^2(\mathcal{P})$:
\begin{gather*}
	X_t =
    A^1_t\frac{\partial}{\partial t} +
    A^2_t\frac{\partial}{\partial s} +
    B^1_t\frac{\partial}{\partial q} +
    B^2_t\frac{\partial}{\partial e} +
    C^1_t\frac{\partial}{\partial p^t} +
    C^2_t\frac{\partial}{\partial p^s},
    \\
	X_s =
    A^1_s\frac{\partial}{\partial t} +
    A^2_s\frac{\partial}{\partial s} +
    B^1_s\frac{\partial}{\partial q} +
    B^2_s\frac{\partial}{\partial e} +
    C^1_s\frac{\partial}{\partial p^t} +
    C^2_s\frac{\partial}{\partial p^s}.
\end{gather*}
Now, Hamilton equations are written as
\begin{equation*}
\begin{cases}
	i(X_t)\,\omega^t + i(X_s)\,\omega^s =
    \d h - \d h(\Reeb_t)\eta^t - \d h(\Reeb_s)\eta^s,
    \\
	i(X_t)\,\eta^t = 1,
    \quad
    i(X_t)\,\eta^s = 0,
    \quad
    i(X_s)\,\eta^t = 0,
    \quad
    i(X_s)\,\eta^s = 1,
\end{cases}
\end{equation*}
which partly determine the coefficients of $\calX$:
\begin{equation*}
\begin{cases}
B^1_t = ep^t,\\
B^1_s = \frac{-1}{\tau}p^s,\\
C^1_t+C^2_s = 0,\\
A^1_t = 1, \quad A^2_t = 0,\\
A^1_s = 0, \quad A^2_s = 1,
\end{cases}
\end{equation*}
and imposes as a consistency condition the secondary Hamiltonian constraint
\begin{equation*}
\xi =
i\left(\frac{\partial}{\partial e}\right)\d h =
\frac{1}{2}(p^t)^2-\frac{1}{2}\sigma^2=
0 \quad
\mbox{\rm (on $\mathcal{P}$)} \:,
\end{equation*}
which defines the submanifold $\mathcal{P}_1\hookrightarrow\mathcal{P}$.
The tangency condition to this new submanifold,
$X_t(\xi)\vert_{\mathcal{P}_1}= 0$,
$X_s(\xi)\vert_{\mathcal{P}_1}= 0$,
determines the coefficients
$\displaystyle
C^1_t\Big\vert_{\mathcal{P}_1}=
\frac{1}{p^t} \,\sigma \frac{\partial \sigma}{\partial t}
$,
$\displaystyle
C^1_s\Big\vert_{\mathcal{P}_1}=
\frac{1}{p^t} \,\sigma \frac{\partial \sigma}{\partial s}
$,
and yields no more constraints.

Observe that $\mathcal{F}L^*(\xi_1)=\zeta_1$,
so $\mathcal{F}L(\mathcal{S}_1)=\mathcal{P}_1$
and, as the semi-holonomy condition does not originate constraints in the Lagrangian formalism,
there are no non-$\mathcal{F}L$-projectable Lagrangian constraints.

\section{Conclusions and outlook}

In this paper the concepts of $k$-precosymplectic manifold and of $k$-precosymplectic Hamiltonian system have been introduced, and we have proved the existence of global Reeb vector fields in these manifolds.
We have developed a constraint algorithm for $k$-precosymplectic (i.e. singular) field theories in order to find a submanifold of the phase bundle where there are solutions to the field equations.
This algorithm can be applied to the Hamiltonian and to the Lagrangian formalisms of these field theories.
In particular, the algorithm allows to find a submanifold where there are multivector fields which are solutions to the geometric field equations, and they are tangent to the submanifold.
These multivector fields are not necessarily integrable on this submanifold, but perhaps on a smaller submanifold of it.

In addition, in the case of the Lagrangian formalism,
the problem of finding {\sc sopde} multivector fields solution to the field equations has been briefly discussed,
but the problem of giving an intrinsic characterization of the ``{\sc sopde} constraints''
is a topic for further research.

Furthermore, an open problem is to find conditions to ensure the existence of some kind of Darboux coordinates in both $k$-presymplectic and $k$-precosymplectic manifolds.
Work on this subject is in progress.

Finally,
the constraint algorithm has been applied to classical field theories
described by affine Lagrangians,
analyzing the Hamiltonian and the Lagrangian formalisms and,
in this last case, the {\sc sopde} condition.
In a future work we would like to apply this analysis
to the study of some models in General Relativity
described by metric-affine Lagrangians,
such as the Einstein--Palatini model.

\subsection*{Acknowledgments}
We acknowledge the financial support from the
Spanish Ministerio de Econom\'{\i}a y Competitividad
project MTM2014--54855--P,
the Ministerio de Ciencia, Innovaci\'on y Universidades project
PGC2018-098265-B-C33,
and the Secretary of University and Research of the Ministry of Business and Knowledge of
the Catalan Government project
2017--SGR--932.
We are indebted to Prof.\ {\it Dieter Van den Bleeken} (Bo\u gazi\c ci University)
for having drawn our attention to one error in the examples of sections \ref{afflag} and \ref{2ndex} 
which has been corrected in the present version.


\end{document}